\documentclass[12pt, onecolumn]{IEEEtran}

\usepackage{cite}
\usepackage{pslatex,bm}
\usepackage{epsfig}
\usepackage{amsmath}
\usepackage{amssymb}
\usepackage{amsthm}
\usepackage{booktabs}
\usepackage{array}
\usepackage{multicol}
\usepackage[usenames, dvipsnames]{color}
\usepackage{url}
\usepackage{dsfont}

\usepackage{threeparttable}
\psfull

\newtheorem{Theo}{Theorem}
\newtheorem{lemma}{Lemma}
\newtheorem{rmk}{\sf Remark}
\newtheorem{definition}{Definition}
\newtheorem{coro}{Corollary}

\DeclareMathAlphabet{\mathpzc}{OT1}{pzc}{m}{it}

\renewcommand{\baselinestretch}{2}

\def\b0{{\mathbf 0}}

\begin{document}

\author{\authorblockN{Pin-Hsun Lin, Carsten Rudolf Janda, and Eduard Axel Jorswieck\\
Communications Laboratory, \\
Department of Electrical Engineering and Information Technology, \\Technische Universit\"{a}t Dresden, Germany\\
Email:{\{pin-hsun.lin, carsten.janda, eduard.jorswieck\}@tu-dresden.de}.
}
\thanks{ This work was supported in part by Fast Cloud and Fast Secure.}
}

\title{Stealthy Secret Key Generation}
\maketitle \thispagestyle{empty}
\begin{abstract}
In this work we consider a \textit{complete} covert communication system, which includes the source-model of a stealthy secret key generation (SSKG) as the first phase. The generated key will be used for the covert communication in the second phase of the current round and also in the first phase of the next round. We investigate the stealthy SK rate performance of the first phase. The derived results show that the SK capacity lower and upper bounds of the source-model SKG are not affected by the additional stealth constraint. This result implies that we can attain the SSKG capacity for free when the sequences observed by the three terminals Alice ($X^n$), Bob ($Y^n$) and Willie ($Z^n$) follow a Markov chain relationship, i.e., $X^n-Y^n-Z^n$. We then prove that the sufficient condition to attain both, the SK capacity as well as the SSK capacity, can be relaxed from physical to stochastic degradedness. In order to underline the practical relevance, we also derive a sufficient condition to attain the degradedness by the usual stochastic order for Maurer's fast fading Gaussian (satellite) model for the source of common randomness.
\end{abstract}

\section{Introduction}

To realize a secure physical layer, concealing the action of transmitting signals from a warden Willie can be seen as a first step to protect information. If the action of the transmission is detected by Willie, the secrecy/confidentiality (or the \textit{hidability} \cite{Che_deniable_com}) provided by the wiretap coding \cite{Wyner_wiretap} can be interpreted as a further protection. There are two main notions to conceal the transmission of signals (or to attain the \textit{deniability} \cite{Che_deniable_com}): 1) \textit{stealthy} communications \cite{Hou_stealth}, \cite{Che_deniable_com}, and 2) \textit{covert} communications/low probability of detection \cite{Bash_LPD13}, \cite{Che_deniable_com, Wang_LPD, Bloch_covert_com16}. Roughly speaking, both notions conceal the desired signal in an ambient signal, such that Willie is not able to distinguish from his observed distributions whether a meaningful transmission is ongoing or not. In particular, in the first notion the meaningful and meaningless signals are transmitted separately in time. Because these two signals have close distributions at Willie, he cannot distinguish them. In contrast, in the second notion the transmitter is either ON or OFF and then the meaningful signal is superimposed on the meaningless one, i.e., the additive noise. The number of messages, which can be covertly transmitted, follows the \textit{square root law}\cite{Bash_LPD13, Che_deniable_com, Wang_LPD, Bloch_covert_com16} of the block length, i.e., the corresponding Shannon's rate is zero. In contrast, a positive capacity can be achievable for the stealth communications. Note that for both notions, if the main (Bob's) channel has no advantage over Willie's channel, additional keys are necessary to conceal the signals, e.g., \cite{Wang_LPD}, \cite{Bloch_covert_com16}. For a more detailed comparison please refer to \cite{Che_arxiv}.

In this paper we aim to design a \textit{complete} covert communication system by investigating the stealthy secret key generation (SSKG) scheme with rate-unlimited public discussions. In each round of such transmission, SSKG in the first phase can provide keys to enable the covert communications in the second phase, when Willie has a stronger channel than Bob \cite{Bloch_covert_com16}. Instead of using normal SKG, the SSKG can avoid Willie's awareness during the first phase. Therefore, combining the first with the second phase, we can attain the complete covert communication. On the contrary, normal SKG schemes may utilize public communications for advantage
distillation, information reconciliation, and privacy amplification \cite{Bloch_book}, which may raise Willie's attention. Therefore, directly applying the normal SKG scheme for the keys violates the constraints of stealthy/covert communications. Note that when the divergence between the distributions of channel outputs for meaningful and meaningless signals at Willie is higher than that at Bob, secret keys shared between Alice and Bob are necessary \cite[Theorem 2]{Bloch_covert_com16} to achieve covert communications. These keys are used to switch between different codebooks to calm Willie.

The reason we consider stealth but not covert communications for the SKG is due to the latency. It is known that covert communications result in a zero rate when the block length approaches infinity (but the number of transmitted messages can be positive) \cite{Bloch_covert_com16}, \cite{Wang_LPD}. Even under finite block length, the extreme low rate due to the covert communications constraint will incur an extremely high latency if, for example, we exploit covert communications in the public discussion to avoid Willie's awareness. 
Further theoretical analysis on the minimum rate of public discussion necessary for maximum SK rate can be referred to \cite{Tyagi_min_pub_discussion_max_SKG}\footnote{Please note that this reference considers the model where Willie observes an independent source to Alice's and Bob's observation.}.
The main contributions of this work are summarized as follows:
\begin{itemize}
\item We consider the stealthy source-model SKG and the \textit{effective secrecy} \cite{Hou_stealth} for the source-model SKG. Based on the  effective secrecy constraint, we derive the capacity lower and upper bounds for SSKG, which correspond to those of SKG without the stealth constraint. It implies that the stealthy SK capacity is unchanged when the common randomness is degraded.
\item We then prove that the sufficient condition to attain the SSK capacity can be relaxed from physical to stochastic degradedness.
\item We also derive a sufficient condition to attain the degradedness by the usual stochastic order \cite{shaked_stochastic_order} for Maurer's fast fading Gaussian (satellite) model \cite{Naito_satellite} for the source of common randomness.
\end{itemize}

The rest of the paper is organized as follows. In Section
\ref{Sec_system_model} we introduce the preliminaries and the
considered system model. In Section \ref{Sec_main_results}
 we derive our main results. In Section \ref{Sec_discussion} we derive the sufficient condition to attain the degradedness. Finally Section
\ref{Sec_conclusion} concludes this paper.\\

\emph{Notation}: Upper case normal/bold letters denote random
variables/random vectors (or matrices), which will be defined
when they are first mentioned; lower case bold letters denote
vectors. And we denote the probability mass function (pmf) by $P$. The entropy of $X$ is defined as $H(X)$. The mutual information between two random variables $X$ and $Y$ is
denoted by $I(X;Y)$. The divergence between distributions $P_X$ and $P_Y$ is denoted by $D(P_X||P_Y)$. The complementary cumulative density function
(CCDF) is denoted by $\bar{F}_X(x)=1-F_X(x)$, where $F_X(x)$ is the
CDF of $X$. The notion
$\,X\sim \,F$ denotes that the random variable $X$ follows the
distribution $F$. The subscript $i$ in $X_i$ denotes the $i$-th symbol and $X^{i}\triangleq[X_1,\,X_2,\,\cdots,\,X_i]$. $X-Y-Z$ denotes the Markov chain. $o(.)$ is the little-o notation for computational complexity. Let $\lceil.\rceil$ denote the ceil operator. All logarithms are with base 2. The stochastic independence between $X$ and $Y$ is denoted by $X\perp\!\!\!\!\perp Y$.

\section{Preliminaries and system model}\label{Sec_system_model}

 From \cite[Theorem 2]{Bloch_covert_com16} we know that given two channels to Bob and Willie, respectively, if Willie can distinguish meaningful and meaningless signals better than Bob, the following amount of secret key bits are necessary to enable covert communications:
\begin{align}\label{EQ_covert_com_R_key}
\log K = \omega_n\sqrt{n}[(1+\xi)D(P_Z||Q_Z)-(1-\xi)D(P_Y||Q_Y)]^+,
\end{align}
where $P_Y$and $Q_Y$ are the output distributions at Bob when Alice is ON and OFF, respectively. Same definitions for $P_Z$ and $Q_Z$ at Willie; $\omega_n\sqrt{n}$ is the number of the transmitted non-innocent symbols within the codeword length $n$ with the constraint $\omega_n=o(1)\cap\omega(1/\sqrt{n})$ when $n\rightarrow\infty$. On the contrary, if Bob can distinguish meaningful or meaningless signals better than Willie in terms of the divergence of these two distributions, no additional keys are needed for covert communications. However, for the former case to generate the key, normal SKG schemes may attract Willie's attention. Therefore, to attain \textit{complete} covert communications, we investigate the SKG with the stealth constraint.
In the considered system, there are two phases for each round of the complete covert communication as shown in Fig. \ref{Fig_packet}. In the first phase, Alice and Bob will use the SSKG to generate keys. In the second phase, part of the generated keys are used for the covert communications to fulfill \eqref{EQ_covert_com_R_key}. Note that the remaining keys will be used for the next round of SSKG, which will be explained later\footnote{A one-time initial key should be shared between Alice and Bob before the whole operation, which is out of the scope of this work.}. In the following we focus on the development of the first phase, where the $n$-time source observations at Alice, Bob, and Willie are denoted by $X^n$, $Y^n$, and $Z^n$, respectively, following the distribution $P_{X^nY^nZ^n}=\prod_{i=1}^nP_{XYZ}$ with alphabets $\mathcal{X},\,\mathcal{Y},\,\mathcal{Z}$, respectively. Denote the public discussion between Alice and Bob by a vector $\bm F\in\mathcal{X}^r$ through a noiseless channel, from which Willie can perfectly observe $\bm F$.

\begin{figure}[h]
\centering \epsfig{file=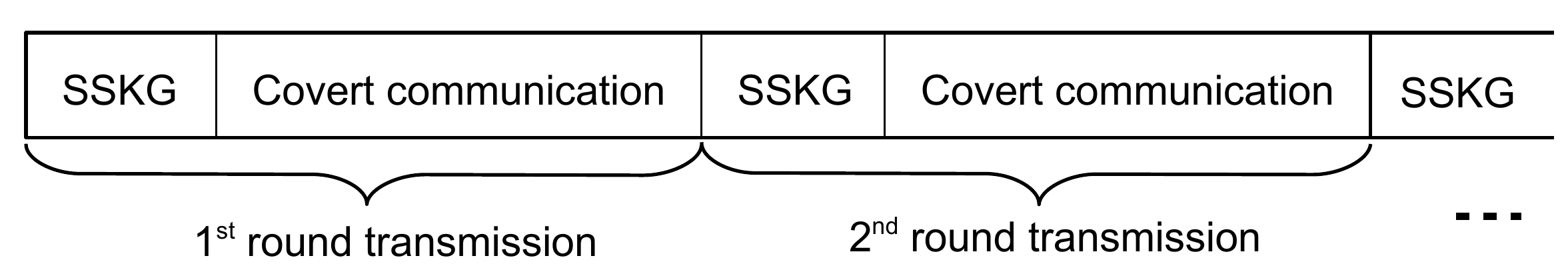, width=0.8\textwidth}
\caption{The proposed transmission scheme for a complete covert communication when Bob has a worse channel than Willie in the sense that $(1-\xi)D(P_Y||Q_Y)<(1+\xi)D(P_Z||Q_Z)$.}
\label{Fig_packet}
\end{figure}

The main step for deriving the SSK capacity lower bound in this paper hinges on constructing a conceptual wiretap channel (CWTC) as in \cite{Maurer_SKG93}, where we set $r=n$. Note that the selection of $r=n$ is due to the construction of the CWTC. In \cite{Ahlswede_SKG93}, the authors proposed a different scheme with $r=1$ to achieve the same SK capacity lower bound. We construct an equivalent wiretap codebook $\{U^n(m,w)\}$, where $m=1,\cdots ,L$ and $w=1,\cdots ,L_1$, $L\triangleq 2^{nR}$ and $L_1\triangleq  2^{nR_1}$ are the numbers of secure and confusion messages, respectively; $m$ and $w$ are uniformly selected, respectively; $U^n(m,w)\in \mathcal{X}^n,\,\forall (m,w)$. In addition, $(Z^n,\bm F,U^n)\sim P_{Z^n\bm FU^n}=\prod_{i=1}^n P_{ZFU}=\prod_{i=1}^n P_{ZF|U}P_U$, where we consider the equivalent channel from Alice to Willie as:
\begin{align}
P_{Z^n\bm F|U^n}&=\prod_{i=1}^n P_{ZF|U},
\end{align}
i.e., the equivalent channel output at Willie is $(Z^n,\,\bm F)$. Similarly, the equivalent channel output at Bob is $(Y^n,\,\bm F)$. Let $U^n$ be independent to $\{X^n,\,Y^n,\,Z^n\}$.
To consider the behavior of stealth, the distributions of the meaningful and meaningless signals at the equivalent channel output at Willie are respectively expressed as:
\begin{align}
P_{Z^n\bm F}&=\frac{1}{LL_1}\sum_{(m,\,w)=(1,1)}^{(L,\,L_1)}P_{Z^n,\bm F|U^n}(z^n,f^n|u^n(m,w)),\label{EQ_code_distr}\\
Q_{Z^n\bm F}&=\sum_{u^n}P_{Z^n,\bm F|U^n}(z^n,f^n|u^n)P_{U^n}(u^n).\label{EQ_target_distr}
\end{align}
In this work, we consider the following constraints:
\begin{align}
 \mathds{P}_e\triangleq\mbox{Pr}(K\neq\hat{K})&\leq \epsilon,\label{EQ_Pe}\\
R-H(K)&\leq \epsilon,\label{EQ_uniformity}\\
D(P_{KZ^n\bm {F}}||P_K P_{Z^n\bm {F}})+D(P_{Z^n\bm {F}}||Q_{Z^n\bm {F}})&\leq\epsilon, \label{EQ_sum_of_constraints}
\end{align}
where \eqref{EQ_Pe} is the average error probability constraint at Bob; \eqref{EQ_uniformity} is the uniformity constraint of the keys; \eqref{EQ_sum_of_constraints} is the effective secrecy constraint, where the first term denotes the \textit{non-confusion} \cite{Hou_stealth} in a strong secrecy manner. In addition, the second term of \eqref{EQ_sum_of_constraints} denotes the \textit{non-stealth}. We can further rearrange \eqref{EQ_sum_of_constraints} by the following:
\begin{align}
&D(P_{KZ^n\bm {F}}||P_K P_{Z^n\bm {F}})+D(P_{Z^n\bm {F}}||Q_{Z^n\bm {F}})\notag\\
=& \sum_{K,Z^n,\bm {F}} P_{KZ^n\bm {F}}\left(\log\frac{P_{KZ^n\bm {F}}}{P_KP_{Z^n\bm {F}}}+\log\frac{P_{Z^n\bm {F}}}{Q_{Z^n\bm {F}}}\right)\notag\\
=& \sum_{K,Z^n,\bm {F}} P_{KZ^n\bm {F}}\left(\log\frac{P_{KZ^n\bm {F}}}{P_KQ_{Z^n\bm {F}}}\right)\notag\\
\triangleq& D( P_{KZ^n\bm {F}}||P_KQ_{Z^n\bm {F}}).\label{EQ_effective_secrecy}
\end{align}
Borrowing the terminology from \cite{Hou_stealth}, we coin \eqref{EQ_effective_secrecy} as the \textit{effective secrecy} for SSKG. The main difference of this work to \cite{Hou_stealth} will be discussed later.

In fact, we can combine the uniformity constraint of the SKG \eqref{EQ_uniformity} with \eqref{EQ_effective_secrecy} as follows:
\begin{align}
D( P_{KZ^n\bm {F}}||P_KQ_{Z^n\bm {F}})+\log|\mathcal{K}|-H(K)
=&\sum P_{KZ^n\bm {F}}\left( \log\frac{P_{KZ^n\bm {F}}}{P_KQ_{Z^n\bm {F}}}\right)-\sum P_{KZ^n\bm {F}}\log\left(\frac{P_{\textsf{Unif}}}{P_K}\right)\notag\\
=&D( P_{KZ^n\bm {F}}||P_{\textsf{Unif}}Q_{Z^n\bm {F}}).
\end{align}

{Note that a common assumption for the public discussion channel is that all parties can access the same discussed signal. To operate under this assumption, additional keys shared between Alice and Bob are necessary for them to distinguish the meaningful signal when Willie is kept unaware, which is one of the main reason that we need the SSKG in the first phase. Otherwise, from channel resolvability it is clear that Bob cannot distinguish whether the received signal is meaningful or meaningless, either. We define the rate of the additional keys as $R_{k, st}$. In the following we derive a SSK rate generated by the SSKG, namely, $R_{SSK}$, which is sufficient to encompass both $R_{k, st}$ (for the first phase in the next round of transmission) and the key rate required by the covert communications (for the second phase in the current round).

\begin{lemma}\label{Lemma_R_SSK}
To achieve a complete covert communication system, the following SK rate $R_{SSK}$ is sufficient
\begin{align}\label{EQ_R_SSK}
R_{SSK}= 1+\frac{o(1)}{\sqrt{n}}\,\,\,bits/channel\,\,use.
\end{align}

\end{lemma}
\begin{proof}
For each discussion it takes at most 1 bit to indicate that it is meaningful or not. Therefore, we set $R_{k, st}=1$ due to the specific use of public channel \cite{Maurer_SKG93}, which results in $r=n$. In addition, the square root law of the key rate required by covert communications \cite[Theorem 2]{Bloch_covert_com16} results in the key rate upper bound $o(1)/\sqrt{n}$.
\end{proof}

\begin{rmk}
The generated SSK bits can be used in different ways for the first phase of each round of transmission. For example, the SSK bit can be used to indicate either each transmitted symbol, or $n$-symbol (i.e., the length of the SSKG), are meaningful or not. Lemma \ref{Lemma_R_SSK} corresponds to the former case, which consumes the largest number of keys. Therefore, Lemma \ref{Lemma_R_SSK} provides us an upper bound of SSK rate for a complete covert communication. On the other hand,
\end{rmk}

\section{Main Result And Proof}\label{Sec_main_results}
\begin{Theo}\label{TH_UB_LB}
The lower and upper bounds of the stealthy secret key capacity $C_{SK}^{Eff}$ of the source-model stealthy secret key generation given a discrete memoryless source $(\mathcal{X},\mathcal{Y},\mathcal{Z},P_{X,Y,Z})$ are
\begin{align}
&\max\{I(X;Y)-I(X,Z),\,I(Y;X)-I(Y;Z)\}\leq C_{SK}^{Eff}\leq \min\{I(X;Y),\,I(X,Y|Z)\}.\label{EQ_C_LB_UB}
\end{align}
\end{Theo}

\begin{rmk}\normalfont
Note that we do not directly apply the effective secrecy \cite{Hou_stealth} which includes both secrecy and stealth constraints, to the public discussion in SKG problems. In contrast, we impose the stealth constraint to the CWTC \cite{Maurer_SKG93} of the SKG and the secrecy constraint is still applied to the source-model SKG.
\end{rmk}

Unlike the wiretap channel with the stealth constraint whose capacity result is shown in \cite{Hou_stealth}, Theorem \ref{TH_UB_LB} only provides the lower and upper bounds. These bounds coincide with those of the secret key capacity without the stealth constraint. However, the same upper and lower bounds of the secret key capacity do not guarantee that the secret key capacity is unchanged when we impose the additional stealth constraint. Therefore, we consider the following case in which the two bounds match. This case leads to the fact that we can get the SSKG for free even with the additional stealth constraint.

\begin{coro}\label{Coro_SK_capacity}
For the discrete memoryless source $(\mathcal{X},\mathcal{Y},\mathcal{Z},P_{XYZ})$, if $X- Y- Z$ forms a Markov chain, then
\begin{align}
C_{SK}^{Eff}=I(X;Y)-I(X;Z).
\end{align}
\end{coro}

\begin{rmk}
By applying the quantization scheme used in \cite[Proof of Theorem 3.3 and
Remark 3.8]{Kim_lecture} or \cite[Appendix B]{Khisti_Secure_fading_BC}, we may extend the SK rate results in Theorem \ref{TH_UB_LB} and Corollary \ref{Coro_SK_capacity} to the Gaussian source.
\end{rmk}
The proofs are derived in the following subsections.

\subsection{Lower Bound of $C_{SK}^{Eff}$} \label{Sec_LB}
To derive the lower bound of $C_{SK}^{Eff}$ in \eqref{EQ_C_LB_UB}, we first decompose the RHS of \eqref{EQ_effective_secrecy} as follows:
\begin{align}\label{EQ_sum_of_constraints3}
D( P_{KZ^n\bm {F}}||P_KQ_{Z^n\bm {F}})&\overset{(a)}=D(P_K||P_K)+D( P_{Z^n\bm {F}|K}||Q_{Z^n\bm {F}}|P_K)\notag\\
&=D( P_{Z^n\bm {F}|K}||Q_{Z^n\bm {F}}|P_K),
\end{align}
where (a) follows the chain rule of divergence from \cite[Th.2.2.2]{Polyanskiy_LNIT}.

Based on the CWTC, we then apply the channel resolvability analysis \cite{Hou_channel_resolvability} to find the rate constraint on $R_1$, i.e., the rate of confusion messages for the codebook generation, which guarantees that the effective secrecy constraint \eqref{EQ_sum_of_constraints} is fulfilled.

From the random coding analysis derived in \ref{APP_Th_LB}, we have:
\begin{align}\label{EQ_LB_result}
\mathds{E}_{\mathcal{C}}[D(P_{Z^n\bm {F}|K}||Q_{Z^n\bm {F}}|P_K)]\leq \mathds{E}_{Z^n\bm {F} U^n}\left[\log\left(\frac{P_{Z^n\bm {F}|U^n}}{L_1Q_{Z^n\bm {F}}}+1\right)\right],
\end{align}
where we recall that $L_1=2^{\lceil nR_1\rceil}$ is the number of confusion message per bin, which is to be designed to guarantee that \eqref{EQ_LB_result} is asymptotically zero. The main difference of this proof to that in \cite{Hou_stealth} is that, by constructing a CWTC for the considered SKG model, we introduce an additional channel output at both Bob and Willie. This difference makes the considered conceptual channel distinct from that in \cite{Hou_stealth}, and those results cannot directly be applied.

To proceed, we reexpress the ratio in the logarithm on the right hand side (RHS) of \eqref{EQ_LB_result} as follows:
\begin{align}
\frac{P_{Z^n\bm {F}|U^n}}{L_1Q_{Z^n\bm {F}}}&\overset{(a)}=\frac{P_{Z^n\bm {F}U^n}}{L_1P_{U^n}}\frac{1}{P_{Z^n}Q_{\bm {F}}}\notag\\
&\overset{(b)}=\frac{P_{Z^n\bm {F}U^n}}{L_1P_{Z^nU^n}}\frac{1}{Q_{\bm {F}}}\notag\\
&=\frac{P_{\bm {F}|Z^nU^n}}{L_1Q_{\bm {F}}},
\end{align}
where (a) is due to the fact that $Z^n$ and $\bm {F}$ are independent when the discussion $\bm F$ is meaningless, whose pmf is denoted by $Q_{\bm F}$; (b) is due to the fact that $U^n$ is selected to be independent to $Z^n$, i.e., $P_{Z^nU^n}=P_{Z^n}P_{U^n}$.

Then we can rewrite \eqref{EQ_LB_result} as follows:
\begin{align}
\mathds{E}_{\mathcal{C}}[D(P_{Z^n\bm {F}|K}||Q_{Z^n\bm {F}}|P_K)]\leq\mathds{E}_{Z^n\bm {F}U^n}\left[\log\left(\frac{P_{\bm {F}|Z^nU^n}}{L_1Q_{\bm {F}}}+1\right)\right].\label{EQ_new_obj}
\end{align}
Similar to \cite{Hou_stealth}, the RHS of \eqref{EQ_new_obj} can be divided into two cases as follows according to whether $(z^n,\,\bm {f},\,u^n)$ are jointly typical or not:
\begin{align}
d_1&= \sum_{\substack{(z^n,\,\bm {f},\,u^n)\in\\ T_{\delta}^n(P_{Z^n,\,\bm {F},\,U^n})}} P_{Z^n \bm {F} U^n}(z^n,\bm {f},u^n)\log\left(\frac{P_{\bm {F}|U^nZ^n}(\bm {f}|u^nz^n)}{L_1Q_{\bm {F}}(\bm {f})}+1\right),\notag\\
d_2&= \sum_{\substack{(z^n,\,\bm {f},\,u^n)\notin\\ T_{\delta}^n(P_{Z^n,\,\bm {F},\,U^n})}} P_{Z^n \bm {F}U^n}(z^n,\bm {f},u^n)\log\left(\frac{P_{\bm {F}|U^nZ^n}(\bm {f}|u^nz^n)}{L_1Q_{\bm {F}}(\bm {f})}+1\right),\notag
\end{align}
where $T_{\delta}^n$ follows the $\delta$-\textit{robust typicality} \cite{Orlitsky_computing} definition for the subsequent derivation. Note that the set of sequences $x^n$ satisfying the definition of robust typicality is denoted by $T_{\delta}^n(P_X)$.

\begin{rmk}\normalfont
Note that even though $Z^n$ and $\bm {F}$ are independent and $Z^n$ and $U^n$ are independent by assumption, that does not mean $Z^n$, $\bm {F}$, and $U^n$ are necessarily generated according to $P_{Z^n,\,\bm {F},\,U^n}=P_{Z^n}P_{\bm {F},\,U^n}$ or $P_{Z^n,\,\bm {F},\,U^n}=P_{Z^n}P_{\bm {F}}P_{U^n}$. In fact, since pairwise independence does not imply mutual independence \cite[Chapter 7.1, 7.2]{Stoyanov_counter_example}, there exists joint distribution $P_{Z^n,\,\bm {F},\,U^n}$ such that we can apply the jointly typical arguments.
\end{rmk}

%

The Chernoff bound and the important upperbound which will be used later are restated in the following.

\begin{lemma}(Chernoff Bound \cite[Lemma 16]{Orlitsky_computing}:) For every $a\in\mathcal{X},$\\
\begin{align}
P\left(\frac{N(a|x^n)}{n}\leq (1+\delta)P_X(a)\right)\leq e^{-\delta^2P_X(a)n/3}.
\end{align}
\end{lemma}

\begin{lemma}(Upper bound of the probability of non-typical set \cite[Lemma 17]{Orlitsky_computing}:)\label{Lemma_UB_non-typical_set}
\begin{align}
P(x^n\notin T_{\delta})\leq 2|S_X|e^{-\delta^2\mu_X n/3},
\end{align}
where $S_X\triangleq \{x\in\mathcal{X}:\,P(x)>0\}$ and $\mu_x\triangleq\min_{x\in S_X}P(x)$.
\end{lemma}
Note that the total rate constraint in the CWTC, i.e., Bob should be able to decode both the secret and confusion messages successfully, which is a point to point transmission problem without secrecy, can be seen from \cite{Orlitsky_computing}. Therefore, we neglect the proof.

Next, we derive the constraint on $R_1$ as follows:
\begin{align}
d_1&\overset{(a)}\leq \Bigg(\sum_{(z^n,\,\bm {f},\,u^n)\in T_{\delta}^n(P_{Z^n\bm {F}U^n})} P_{Z^n\bm {F}U^n}(z^n,\bm {f},u^n)\Bigg)\log\left(\frac{2^{-n[H(F|UZ)-\delta]}}{L_1 2^{-n(1+\epsilon)H(F)}}+1\right)\notag\\
&\overset{(b)}\leq \log\left(\frac{2^{-n[H(F|UZ)-\delta]}}{L_1 2^{-n(1+\epsilon)H(F)}}+1\right)\notag\\
&\overset{(c)}= \log\left(2^{-n\left(R_1-I(F;UZ)-\epsilon'\right)}+1\right),\label{EQ_d1}
\end{align}
where (a) is by \cite[Lemma 18, Lemma 20]{Orlitsky_computing} for the typicality and conditional typicality bounds; (b) is by the fact that the sum probability of jointly typical set is less than 1; (c) is by the definition of $L_1$ and $\epsilon'\triangleq \epsilon(1+H(U))$. Then we know that $d_1\rightarrow 0$ when $n\rightarrow\infty$ if
\begin{align}\label{EQ_R_confusion}
R_1>I(F;UZ)\overset{(a)}=I(U\oplus X;UZ)\overset{(b)}=H(U)-H(X|Z),
\end{align}
where (a) is by the specific use of the public discussion according to \cite[Theorem 3]{Maurer_SKG93}, $\oplus$ is the modulo addition in $\mathcal{X}$; (b) is due to the fact that $U$ is uniformly distributed followed by the crypto lemma.
In addition, we can derive that $d_2\rightarrow 0$ as $n\rightarrow\infty$ as follows:
\begin{align}
d_2&\overset{(a)} \leq \sum_{(z^n,u^n,\,\bm {f})\notin T_{\delta}^n(P_{Z^n,U^n,\bm {F}})} P_{Z^n \bm {F}U^n}(z^n,\bm {f},u^n)\log\left(\frac{1}{Q_{\bm {F}}(\bm {f})}+1\right)\notag\\
&\overset{(b)} \leq \sum_{(z^n,u^n,\,\bm {f})\notin T_{\delta}^n(P_{Z^n,U^n,\bm {F}})} P_{Z^n \bm {F}U^n}(z^n,\bm {f},u^n)\log\left(\frac{1}{\mu_{ {f}}}+1\right)\notag\\
&\overset{(c)}= P((z^n,u^n,\,\bm {f})\notin T_{\delta}^n(P_{Z^n,U^n,\bm {F}})) \log\left(\frac{1}{\mu_{f}}+1\right)\notag\\
&\overset{(d)}\leq 2|S_{ZUF}|e^{-\delta^2\mu_{ZUF} n/3}\log\left(\frac{1}{\mu_{f}}+1\right),\label{EQ_d2}
\end{align}
where (a) is due to the fact that $P_{\bm {F}|U^nZ^n}(\bm {f}|u^nz^n)\leq 1$ and $L_1>1$, and therefore, removing $P_{\bm {F}|U^nZ^n}(\bm {f}|u^nz^n)/L_1$ will upper bound $d_2$; (b) is by lower bounding $Q_{\bm {F}}^n(\bm {f})$ with $\mu_{f}=\min_{\bm {f}\in S_{\bm {F}}}Q_{\bm {F}}(\bm {f})$, where $S_{\bm {F}}\triangleq \{Q_{\bm {F}}\in\mathcal{X}^r: P(\bm {f})>0\}$; (c) is by definition of probability; (d) is by Lemma \ref{Lemma_UB_non-typical_set}. From \eqref{EQ_d2} it can be easily seen that if $n\rightarrow\infty$, $d_2\rightarrow 0$ exponentially fast.

Then from \eqref{EQ_d1} and \eqref{EQ_d2} it is clear that \eqref{EQ_sum_of_constraints} is fulfilled.

From the CWTC construction we know that the following rate between Alice and Bob is achievable:
\begin{align}
n(R+R_1)&\leq I(U^n;U^n\oplus X^n,Y^n)\notag\\
&= H(U^n\oplus X^n,Y^n)-H(U^n\oplus X^n,Y^n|U^n)\notag\\
&\overset{(a)}= H(U^n)+H(Y^n)-H(X^n,Y^n)\notag\\
&= H(U^n)-H(X^n|Y^n),\label{EQ_sum_rate_constraint}
\end{align}
where (a) is due to the crypto lemma and the selection of $U^n$ is independent to $Y^n$. Then from \eqref{EQ_R_confusion} and \eqref{EQ_sum_rate_constraint}, we can derive the achievable SSK rate as follows:
\begin{align}
nR&\leq H(U^n)-H(X^n|Y^n) -nR_1 \notag\\
&\overset{(a)}< n[H(X|Z)-H(X|Y)]\notag\\
&= n[I(X;Y)-I(X;Z)],
\end{align}
where (a) is by substituting \eqref{EQ_R_confusion} in addition to the assumption of memoryless and independent and identically distributed (i.i.d.) common randomness.
Due to symmetry between Alice and Bob, their role can be exchanged and the other lower bound derived. This completes the proof.

Note that from the chain rule of the divergence we know that
\begin{align}\label{EQ_DPI_stealth_constraints}
D(P_{Z^n\bm F}||Q_{Z^n\bm F})=D(P_{\bm F}||Q_{\bm F})+D(P_{Z^n|\bm F}||Q_{Z^n|\bm F}|P_{\bm F}).
\end{align}
Since the left hand side of \eqref{EQ_DPI_stealth_constraints} is constrained by \eqref{EQ_sum_of_constraints} and the conditional divergence is nonnegative, we know that the effective secrecy of the SSKG implies $D(P_{\bm F}||Q_{\bm F})\leq \epsilon$.

\begin{rmk}
When applying the crypto lemma in \eqref{EQ_R_confusion} or \eqref{EQ_sum_rate_constraint} for unbounded $X$, e.g., Gaussian cases, we may follow the argument in \cite[Appendix B]{Bennatan_infinite_modulo}. In particular, a mutual information gap $\delta_1$ can be introduced. Note that $\delta_1\rightarrow 0$ when the modulo size approaches infinity.
\end{rmk}

\subsection{Upper Bound of $C_{SK}^{Eff}$}
In this subsection we derive the upper bound of $C_{SK}^{Eff}$ as follows, which is mainly adapted from the normal steps to derive the upper bound of source-model SKG, e.g., \cite[Sec. 4.2.1]{Bloch_book}, with modification to encompass the effective secrecy constraint:
\begin{align}
nR&\leq\log\lceil2^{nR}\rceil\overset{(a)}\leq H(K)+\epsilon\notag\\
&\overset{(b)}\leq H(K|\bm {F},Z^n)+D( P_{KZ^n\bm {F}}||P_KP_{Z^n\bm {F}})+D(P_{Z^n\bm {F}}||Q_{Z^n\bm {F}})+\epsilon\notag\\
&\overset{(c)}\leq H(K|\bm {F},Z^n)+2\epsilon\notag\\
&= I(K;\hat{K}|\bm {F},Z^n)+H(K|\hat{K},\bm {F},Z^n)+2\epsilon\notag\\
&\overset{(d)}\leq I(K;\hat{K}|\bm {F},Z^n)+\epsilon_2\notag\\
&\overset{(e)}\leq I(X^nR_X\bm {F}^B;Y^nR_Y\bm {F}^A|\bm {F},Z^n)+\epsilon_2\notag\\
&\overset{}\leq I(X^nR_X;Y^nR_Y|\bm {F},Z^n)+\epsilon_2\notag\\
&\overset{(f)}\leq I(X^nR_X;Y^nR_Y|Z^n)+\epsilon_2\notag\\
&\overset{(g)}= I(X^n;Y^n|Z^n)+\epsilon_2\notag\\
&\overset{(h)}= nI(X;Y|Z)+\epsilon_2,\notag
\end{align}
where (a) is by \eqref{EQ_uniformity}; (b) is by definition of divergence and the fact that divergence is positive; (c) is from \eqref{EQ_sum_of_constraints}; (d) is due to Fano's inequality: $H(K|\hat{K},\bm {F},Z^n)\leq\epsilon_1$ and by defining $\epsilon_2\triangleq 2\epsilon+\epsilon_1$; (e) follows the chain rule $K- X^nR_X\bm {F}^B- Y^nR_Y\bm {F}^A- \hat{K}$, where $R_X$ and $R_Y$ are the local randomness, $\bm {F}^A$ and $\bm {F}^B$ are the discussion signals sent by Alice and Bob, respectively, and $\bm F=(\bm F^A,\,\bm F^B)$; (f) is due to \cite[Lemma 4.2]{Bloch_book}; (g) is due to the fact that the local randomness $(R_X,R_Y)$ is selected to be independent to $(X^n,Y^n,Z^n)$; (h) follows from the fact that $(X,Y,Z)$ is a memoryless source.

Following the same steps, we can derive another upper bound without conditioning on $Z$, which completes the proof.

\begin{rmk}\normalfont
Other tighter outer bounds derived by, e.g., the \textit{intrinsic conditional information} \cite[P. 130]{Bloch_book} and \textit{reduced intrinsic conditional information} \cite[P. 133]{Bloch_book} can be proved unchanged even when the stealthy public discussion is considered. This is because that those derivation is irrelevant to the stealth constraint.
\end{rmk}

\section{On the Sufficient Conditions for Degraded Common Randomness}\label{Sec_discussion}
In the following, we prove that the sufficient condition to achieve $C_{SK}^{Eff}=I(X;Y)-I(X;Z)$, i.e., the common randomness forming a Markov chain $X-Y-Z$, which is physically degraded, can be relaxed to be stochastically degraded.
We then show that the relaxed condition can be fulfilled in a broader sense by considering Maurer's fast fading Gaussian (satellite) model \cite{Naito_satellite}.
More specifically, there exists a central random source $S$ passing through fast fading additive white Gaussian noise (AWGN) channels and then observed as $X$, $Y$, and $Z$ at Alice, Bob, and Willie, respectively.
We apply the usual stochastic order \cite{shaked_stochastic_order} to derive a sufficient condition on the fading channels such that $C_{SK}^{Eff}$ is achieved.
The derived sufficient condition provides a simple way to verify the stochastic degradedness and thereby to identify the effective SK capacity. We first give the definition on the degraded relation between the common randomness followed by our result.
\begin{definition}\normalfont\label{Def_PHY-Sto_D}
A source of common randomness $(\mathcal{X},\mathcal{Y},\mathcal{Z},P_{X\tilde{Y}\tilde{Z}})$ is called stochastically degraded if the conditional marginal distributions $P_{\tilde{Y}|X}$ and $P_{\tilde{Z}|X}$ are identical to those of another source of common randomness $(\mathcal{X},\mathcal{Y},\mathcal{Z},P_{XYZ})$ following the physical degradedness, i.e., $X-Y-Z$.
\end{definition}
\begin{Theo}\label{Th_SK_capacity_condition}
If a source of common randomness $(\mathcal{X},\mathcal{Y},\mathcal{Z},P_{X\tilde{Y}\tilde{Z}})$ is stochastically degraded such that $P_{\tilde{Y}|X}=P_{Y|X}$ and $P_{\tilde{Z}|X}=P_{Z|X}$, where $X-Y-Z$, then $C_{SK}^{Eff}=I(X;Y)-I(X;Z)$.
\end{Theo}
\begin{proof}
We prove that the stochastically degraded source $(X,\tilde{Y},\tilde{Z})$ implies that the corresponding CWTC is also stochastically degraded.
This implies that $C_{SK}^{Eff}$ is the same as that of the CWTC from a physically degraded source of common randomness by the same marginal property of WTC.
We start from checking the CWTC of the source $(X,Y,Z)$, where the equivalently received signals at Bob and Willie are $Y'\triangleq (Y,U\oplus X)$ and $Z'\triangleq (Z,U\oplus X)$, respectively.
If $X-Y-Z$, then $U-Y'-Z'$, i.e., the CWTC is a physically degraded one, which can be proved by showing $I(U;Z'|Y')=0$ as follows:
\begin{align}
I(U;Z'|Y')&=H(Z'|Y')-H(Z'|Y',U)\notag\\
&\overset{(a)}=H(Z,U\oplus X|Y,U\oplus X)-H(Z,U\oplus X|Y,U\oplus X,U)\notag\\
&\overset{(b)}=H(U\oplus X|Y,U\oplus X)+H(Z|Y,U\oplus X)-H(Z,U\oplus X|Y,U\oplus X,U)\notag\\
&=H(Z|Y,U\oplus X)-H(Z,U\oplus X|Y,U\oplus X,U)\notag\\
&\overset{(c)}=H(Z|Y)-H(Z,U\oplus X|Y,U\oplus X,U)\notag\\
&\overset{(d)}=H(Z|Y)-H(U\oplus X|Y,U\oplus X,U)-H(Z|Y,U\oplus X,U)\notag\\
&=H(Z|Y)-H(Z|Y,U\oplus X,U)\notag\\
&\overset{(e)}=H(Z|Y)-H(Z|Y, X,U)\notag\\
&\overset{(f)}=H(Z|Y)-H(Z|Y, X)\notag\\
&=I(X;Z|Y)\notag\\
&\overset{(g)}=0,
\end{align}
where (a) is by definition of $Y'$ and $Z'$; (b) is by the chain rule of entropy; (c) is by the crypto lemma and $U$ is selected to be independent to $Y$ and $Z$; (d) is again by the chain rule of entropy; (e) is from the fact that given $U$, we can know $X$ from $U\oplus X$; (f) is by again by the selection that $U$ is selected to be independent to $X$, $Y$ and $Z$; (g) is due to $X-Y-Z$.
Due to the CWTC, we can invoke the same marginal property \cite[Lemma 2.1]{Liang_security_book}: if there exist other equivalent channel outputs $Y''\triangleq (\tilde{Y},U\oplus X)$ and $Z''\triangleq (\tilde{Z},U\oplus X)$ at Bob and Willie, respectively, and if $P_{Y''|U}=P_{Y'|U}$ and $P_{Z''|U}=P_{Z'|U}$, then the two WTCs have the same capacity-equivocation region.
Since
\begin{align}\label{EQ_Y_prime_YU}
P_{Y'}=P_{YU\oplus X}\overset{(a)}=P_{YU}\overset{(b)}=P_YP_U,
\end{align}
where (a) is due to the crypto lemma and (b) is by the selection of $U$ to be independent to the common randomness, we then have $P_{Y'|U}=P_Y$.
Similarly, we have $P_{Z'|U}=P_Z$, $P_{Y''|U}=P_{\tilde{Y}}$, and $P_{Z''|U}=P_{\tilde{Z}}$.
If $(U,Y'',Z'')$ forms a stochastically degraded WTC corresponding to the physically degraded WTC $(U,Y',Z')$, from \cite[Lemma 13.16]{Moser_adv} we have
\begin{align}
P_{Z''|U}(z|u)=\sum_{y} P_{Y''|U}(y|u)P_{Z'|Y'}(z|y)\overset{(a)}=\sum_{y} P_{Y}(y)P_{Z'|Y'}(z|y),
\end{align}
where (a) is by $P_{Y''|U}=P_{Y'|U}=P_Y$. Therefore, by $P_{Z''|U}=P_{Z'|U}$, we have
\begin{align}\label{EQ_Chain_CWTC}
P_{\tilde{Z}}(z)=\sum_{y} P_{Y}(y)P_{Z'|Y'}(z|y).
\end{align}
Now consider the stochastically degraded source of common randomness $(X, \tilde{Y}, \tilde{Z})$ fulfilling $P_{\tilde{Y}|X}=P_{Y|X}$ and $P_{\tilde{Z}|X}=P_{Z|X}$.
Similar to CWTC, we consider the following property for the stochastically degraded source $(X, \tilde{Y}, \tilde{Z})$ according to Definition \ref{Def_PHY-Sto_D}:
\begin{align}\label{EQ_intermediate_chain_source}
P_{\tilde{Z}|X}(z|x)=\sum_{y} P_{\tilde{Y}|X}(y|x)P_{Z|Y}(z|y)\overset{(a)}=\sum_{y} P_{{Y}|X}(y|x)P_{Z|Y}(z|y),
\end{align}
where (a) is by $P_{\tilde{Y}|X}=P_{Y|X}$. After marginalization over $X$ on both sides of \eqref{EQ_intermediate_chain_source}, we get
\begin{align}\label{EQ_chain_source}
P_{\tilde{Z}}(z)=\sum_{y} P_{Y}(y)P_{Z|Y}(z|y).
\end{align}
In addition, we can derive that
\begin{align}\label{EQ_chain_equivalence}
P_{Z'|Y'}&=\frac{P_{Z'Y'}}{P_{Y'}}\overset{(a)}=\frac{P_{YZX\oplus U}}{P_YP_U}\notag\\
&\overset{(b)}=\frac{P_{YZU}}{P_YP_U}\notag\\
&\overset{(c)}=\frac{P_{YZ}P_{U}}{P_YP_U}=P_{Z|Y},
\end{align}
where (a) is by definitions of $Y'$ and $Z'$ and due to \eqref{EQ_Y_prime_YU}; (b) is by the crypto lemma; (c) is by the selection of $U$. From \eqref{EQ_chain_equivalence} we know that the expressions of the stochastic degradedness of the source and CWTC, i.e., \eqref{EQ_Chain_CWTC} and \eqref{EQ_chain_source}, are the same. Then it follows that the stochastically degraded $(X,\tilde{Y},\tilde{Z})$ implies that $(U,Y'',Z'')$ is also stochastically degraded, vice versa. In addition, by the same marginal property, the WTCs formed by $(U, Y', Z')$ and $(U, Y'',Z'')$ have the same secrecy capacity, which completes the proof.
\end{proof}

From Theorem \ref{Th_SK_capacity_condition} we can have the following observation.

\begin{coro}
The lower bound in Theorem \ref{TH_UB_LB} is tight for stochastically degraded source of common randomness $(X,\,\tilde{Y},\,\tilde{Z})$.
\end{coro}
\begin{proof}
Due to same marginal property, we have $P_{\tilde{Y}|X}=P_{Y|X}$ and $P_{\tilde{Z}|X}=P_{Z|X}$, which imply $P_{\tilde{Y}}=P_{Y}$ and $P_{\tilde{Z}}=P_{Z}$. Then by definition of mutual information, we can easily see that
\begin{align}
I(X;\tilde{Y})-I(X;\tilde{Z})=I(X;Y)-I(X;Z).
\end{align}
\end{proof}
\begin{rmk}
However, the upper bound in Theorem \ref{TH_UB_LB} cannot be tight when Theorem \ref{Th_SK_capacity_condition} is valid. This is because
\begin{align}
I(X;Y'|Z')&=I(X;Y'Z')-I(X;Z')\notag\\
&=I(X;Y')-I(X;Z')+I(X;Z'|Y')\notag\\
&\overset{(a)}=I(X;Y)-I(X;Z)+I(X;Z'|Y'),\notag
\end{align}
where (a) is by the same marginal property. Note that $I(X;\tilde{Z}|\tilde{Y})$ cannot be zero since $I(X;\tilde{Z}|\tilde{Y})=0$ if and only if $X-\tilde{Y}-\tilde{Z}$ \cite[Theorem 2.5]{Polyanskiy_LNIT}. But here there is no such Markov chain $X-\tilde{Y}-\tilde{Z}$.
\end{rmk}

In the following we give an example scenario of Theorem \ref{Th_SK_capacity_condition}.

Example 1: Consider Maurer's fast fading Gaussian (satellite) model \cite{Naito_satellite} as follows:
\begin{align}
 X&=A_XS,\notag\\
 Y&=X+N_Y=A_XS+N_Y,\notag\\
 Z&=A_ZS+N_Z ,
\end{align}
  where $N_Y$ and $N_Z$ are independent AWGNs at Bob and Willie, respectively, while both are with zero mean and unit variance; $A_X$ and $A_Z$ follow CDFs $F_X$ and $F_Z$, respectively, are the i.i.d. fast fading channel gains from the source $S$ to Alice and Willie, respectively. Note that $Y$ and $Z$ have no degradedness relation in general due to the random fading. Commonly, we only consider deterministic channel gains with the order $a_X^2\geq a_Z^2$ to form the stochastic degradedness, where $a_X$ and $a_Z$ are realizations of $A_X$ and $A_Z$, respectively. However, the following result broadens the scenarios to get the degradedness among different observations of the same source.

\begin{Theo}
If the random channels $A_X$ and $A_Z$ fulfill $\bar{F}_{{A_X^2}}(x)\geq\bar{F}_{{A_Z^2}}(x)$ for all $x$, where the subscripts denote the absolute square of the channels, then $(X,Y,Z)$ is equivalent to the observations of a source $(\hat{X},\hat{Y},\hat{Z})$, which is degraded, where $\hat{X}=\hat{A}_X S$, $\hat{Y}=\hat{A}_X S+N_Y$, $\hat{Z}=\hat{A}_Z S+N_Z$, $\hat{A}_X^2=F_{A_X^2}^{-1}(U)$, $\hat{A}_Z^2=F_{A_Z^2}^{-1}(U)$, $U\sim$Unif(0,1), $U\perp\!\!\!\!\perp\{N_Y,N_Z,S\}$.
\end{Theo}
\begin{proof}
To proceed, we first introduce the following definition and theorem.
\begin{definition}\normalfont\cite[(1.A.3)]{shaked_stochastic_order}\label{shaked_stochastic_order}
For random variables $A$ and $B$, $A\leq_{st} B$ if and only if
$\bar{F}_{A}(a)\leq \bar{F}_{B}(a)$ for all $a$.
\end{definition}
Let $A=_{st}A'$ denote that $A$ and $A'$ have the same distribution.
\begin{Theo}\label{Th_Strassen}\normalfont Coupling \cite{Thorisson_coupling}:
 $A\leq_{st} B$ if and only if there exists random variables $\hat{A}=_{st} A$ and $\hat{B}=_{st} B$ such that $\hat{A}\leq \hat{B}$ almost surely.
\end{Theo}
Therefore, from Theorem \ref{Th_Strassen} we have observations at Bob and Willie as $\hat{Y}=\hat{A}_XS+N_Y$ and $\hat{Z}=\hat{A}_ZS+N_Z$, respectively, where $\hat{A}_X^2\geq \hat{A}_Z^2$ almost surely, and $F_{A_X^2}(x)=F_{\hat{A}_X^2}(x)$ and $F_{A_Z^2}(x)=F_{\hat{A}_Z^2}(x)$, for all $x$. Similar to the proof steps in Theorem \ref{Th_SK_capacity_condition}, by the same marginal property when considering the CWTC, $(\hat{X},\hat{Y},\hat{Z})$ form equivalently stochastically degraded observations to the original ones $(X,Y,Z)$ in the sense of having the same SK capacity. Therefore, it is clear that $\bar{F}_{{A_X^2}}(x)\geq\bar{F}_{{A_Z^2}}(x)$ is a relaxed sufficient condition to guarantee that $Z$ is an equivalently stochastically degraded version of $Y$. The equivalent channels can be explicitly constructed as $\hat{A}_X^2=F_{A_X^2}^{-1}(U)$ and $\hat{A}_Z^2=F_{A_Z^2}^{-1}(U)$, is according to, e.g., the proof of \cite[Proposition 2.3]{Ross_2nd_course}.
\end{proof}

Example 2: Continuing Example 1, assume $A_X$ and $A_Z$ are from fading channels with their magnitudes following Nakagami-$m$ distribution with \textit{shape
parameters} $m_x$ and $m_z$, and \textit{spread parameters} $w_x$ and $w_z$ \cite{Simon_fading}, respectively. From Theorem \ref{Th_SK_capacity_condition} we know that $Z$ is a degraded version of $Y$ if
\begin{align}
\frac{\gamma\left(m_x,\frac{m_x }{w_x}x\right)}{\Gamma(m_x)}\geq \frac{\gamma\left(m_z,\frac{m_z }{w_z}x\right)}{\Gamma(m_z)},\,\forall x,\notag
\end{align}
where $\gamma(s,x)=\int_0^xt^{s-1}e^{-t}dt$ is the incomplete gamma function and $\Gamma(s)=\int_0^{\infty}t^{s-1}e^{-t}dt$ is the ordinary gamma function. An example satisfying the above inequality is $(m_x,w_x)=(1,3)$ and $(m_z, w_z)=(1,2)$.

%

\section{conclusion}\label{Sec_conclusion}
In this work we investigate the effect of stealthy public discussion used in the source-model of secret key generation. The results show that the SK capacity lower and upper bounds of the source-model are not affected by the additional stealth constraint. This implies that we can attain the stealthy SK capacity for free when the common randomness forms a Markov chain. We then prove that the sufficient condition to attain the SK capacity can be relaxed from physical to stochastic degradedness. We also derive a sufficient condition to attain the degradedness by the usual stochastic order for Maurer's fast fading Gaussian (satellite) model for the common randomness source.

\newpage
\renewcommand{\thesection}{Appendix I}
\section{Proof of \eqref{EQ_LB_result}}\label{APP_Th_LB}
\begin{align}
&\mathds{E}_{\mathcal{C}}[D(P_{Z^n\bm F|K}||Q_{Z^n\bm F}|P_K)]\notag\\
&\overset{(a)}=\mathds{E}_{\mathcal{C}}[D(P_{Z^n\bm F|M}||Q_{Z^n\bm F}|P_M)]\notag\\
&\overset{(b)}=\mathds{E}_{\mathcal{C},M}\left[\sum_{z^n,\bm F}P_{Z^n,\bm F|M}(z^n,\bm f|M=m)\log\left(\frac{P_{Z^n,\bm F|M}(z^n,\bm f|M=m)}{Q_{Z^n,\bm F}(z^n,\bm f)}\right)\right]\notag\\
&\overset{(c)}=\mathds{E}_{\mathcal{C},M}\left[\sum_{z^n,\bm F}\sum_{w=1}^{L_1}\frac{1}{L_1}P_{Z^n,\bm F|U^n}(z^n,\bm f|u^n(m,w))\log\left(\frac{\sum_{l=1}^{L_1}\frac{1}{L_1}P_{Z^n,\bm F|U^n}(z^n,\bm f|u^n(m,l))}{Q_{Z^n,\bm F}(z^n,\bm f)}\right)\right]\notag\\
&\overset{(d)}=\frac{1}{LL_1}\mathds{E}_{\mathcal{C}}\left[\sum_{z^n,\bm F}\sum_{m=1}^{L}\sum_{w=1}^{L_1}P_{Z^n,\bm F|U^n}(z^n,\bm f|u^n(m,w))\log\left(\frac{\sum_{l=1}^{L_1}Q_{Z^n,\bm F|U^n}(z^n,\bm f|u^n(m,l))}{L_1P_{Z^n,\bm F}(z^n,\bm f)}\right)\right]\notag\\
&\overset{(e)}=\frac{1}{LL_1}\mathds{E}_{\mathcal{C}}\left[\sum_{z^n,\bm F}\sum_{m=1}^{L}\sum_{w=1}^{L_1}P_{Z^n,\bm F|U^n}(z^n,\bm f|u^n(m,w))\log\left(\frac{\vartriangle_m(z^n,\bm f|u^n)}{\square(z^n,\bm f)}\right)\right]\notag\\
&\overset{(f)}=\frac{1}{LL_1}\sum_{u^n(1,1)}\cdots\sum_{u^n(L,L_1)}\prod_{m=1,w=1}^{L,L_1} P_U^n(u^n(m,w))\left[\sum_{z^n,\bm F}\sum_{m=1}^{L}\sum_{w=1}^{L_1}P_{Z^n,\bm F|U^n}(z^n,\bm f|u^n)\log\left(\frac{\vartriangle_m(z^n,\bm f|u^n)}{\square(z^n,\bm f)}\right)\right]\notag\\
&\overset{(g)}=\frac{1}{LL_1}\sum_{u^n(1,1)}\cdots\sum_{u^n(L,L_1)}\prod_{m=1,w=1}^{L,L_1} P_U^n(u^n(m,w))\notag\\
&\sum_{z^n,\bm F}
\left[
 \begin{array}{ccc}
     (P_{Z^n,\bm F|U^n}(\cdot|u^n(1,1))+ & \cdots  & +P_{Z^n,\bm F|U^n}(\cdot|u^n(1,L_1)))\cdot\log\left(\frac{\vartriangle_1(z^n,\bm f|u^n)}{\square(z^n,\bm f)}\right)+ \\
     \vdots & \ddots & \vdots \\
     (P_{Z^n,\bm F|U^n}(\cdot|u^n(L,1)) & \cdots  & +P_{Z^n,\bm F|U^n}(\cdot|u^n(L,L_1)))\cdot\log\left(\frac{\vartriangle_L(z^n,\bm f|u^n)}{\square(z^n,\bm f)}\right) \\
   \end{array}
 \right]\notag\\
&\overset{(h)}=\frac{1}{LL_1}\sum_{u^n(1,1)}\cdots\sum_{u^n(L,L_1)}\notag\\
&\sum_{z^n,\bm F}\left[
 \begin{array}{c}
    \Big(P_{Z^n,\bm F,U^n}(\cdot,u^n(1,1))\underset{m\neq 1,w\neq 1}{\overset{L,L_1} \prod} P_U^n(u^n(m,w))+  \cdots   +P_{Z^n,\bm F,U^n}(\cdot,u^n(1,L_1))\underset{m\neq 1,w\neq L_1}{\overset{L,L_1} \prod} P_U^n(u^n(m,w))\Big)\cdot\notag\\
    \log\left(\frac{\vartriangle_1(z^n,\bm f|u^n)}{\square(z^n,\bm f)}\right)+\cdots \\
    \Big(P_{Z^n,\bm F,U^n}(\cdot,u^n(L,1))\underset{m\neq L,w\neq 1}{\overset{L,L_1} \prod} P_U^n(u^n(m,w))+   \cdots   +P_{Z^n,\bm F,U^n}(\cdot,u^n(L,L_1))\underset{m\neq L,w\neq L_1}{\overset{L,L_1} \prod} P_U^n(u^n(m,w))\Big)\cdot\notag\\
    \log\left(\frac{\vartriangle_L(z^n,\bm f|u^n)}{\square(z^n,\bm f)}\right) \\
   \end{array}
 \right]\notag
\end{align}
\begin{align}
&\overset{(i)}=\frac{1}{LL_1}\sum_{z^n,\bm F}\left[
 \begin{array}{c}
    \underset{u^n(1,1)}\sum P_{Z^n,\bm F,U^n}(\cdot,u^n(1,1))\mathds{E}_{\underline{U}^n\setminus\underline{U}^n(1,1)}\left[\log\left(\frac{\vartriangle_1(z^n,\bm f|u^n)}{\square(z^n,\bm f)}\right)\right]+ \cdots + \\
    \underset{u^n(k,l)}\sum P_{Z^n,\bm F,U^n}(\cdot,u^n(k,l))\mathds{E}_{\underline{U}^n\setminus\underline{U}^n(k,l)}\left[\log\left(\frac{\vartriangle_k(z^n,\bm f|u^n)}{\square(z^n,\bm f)}\right)\right]+ \cdots + \\
    \underset{u^n(L,L_1)}\sum P_{Z^n,\bm F,U^n}(\cdot,u^n(L,L_1))\mathds{E}_{\underline{U}^n\setminus\underline{U}^n(L,L_1)}\left[\log\left(\frac{\vartriangle_L(z^n,\bm f|u^n)}{\square(z^n,\bm f)}\right)\right] \\
   \end{array}
 \right]\notag\\
&\overset{(j)}=\frac{1}{LL_1}\sum_{z^n,\bm F}\underset{(a,b)=(1,1)}{\overset{(L,L_1)}\sum}\underset{u^n(a,b)}\sum P_{Z^n,\bm F,U^n}(\cdot,u^n(a,b))\mathds{E}_{\underline{U}^n\setminus\underline{U}^n(a,b)}\left[\log\left(\frac{\vartriangle_a(z^n,\bm f|u^n)}{\square(z^n,\bm f)}\right)\right]\notag\\
&\overset{(k)}\leq\frac{1}{LL_1}\sum_{z^n,\bm F}\underset{(a,b)=(1,1)}{\overset{(L,L_1)}\sum}\underset{u^n(a,b)}\sum P_{Z^n,\bm F,U^n}(\cdot,u^n(a,b))\log\left(\mathds{E}_{\underline{U}^n\setminus\underline{U}^n(a,b)}\left[\left(\frac{\vartriangle_a(z^n,\bm f|u^n)}{\square(z^n,\bm f)}\right)\right]\right)\notag\\
&\overset{(l)}=\frac{1}{LL_1}\sum_{z^n,\bm F}\underset{(a,b)=(1,1)}{\overset{(L,L_1)}\sum}\underset{u^n(a,b)}\sum P_{Z^n,\bm F,U^n}(\cdot,u^n(a,b))\log\left(\frac{P_{Z^n,\bm F|U^n}(\cdot|u^n(a,b))+\underset{s\neq b}\sum \underset{u^n(a,s)}\sum P_{Z^n,\bm F,U^n}(\cdot,u^n(a,s))}{\square(z^n,\bm f)}\right)\notag\\
&\overset{(m)}\leq\frac{1}{LL_1}\sum_{z^n,\bm F}\underset{(a,b)=(1,1)}{\overset{(L,L_1)}\sum}\underset{u^n(a,b)}\sum P_{Z^n,\bm F,U^n}(\cdot,u^n(a,b))\log\left(\frac{P_{Z^n,\bm F|U^n}(\cdot|u^n(a,b))+\underset{(r,s)=(1,1)}{\overset{(L,L_1)}\sum} \underset{u^n(r,s)}\sum P_{Z^n,\bm F,U^n}(\cdot,u^n(r,s))}{\square(z^n,\bm f)}\right)\notag\\
&\overset{(n)}\leq\frac{1}{LL_1}\sum_{z^n,\bm F}\underset{(a,b)=(1,1)}{\overset{(L,L_1)}\sum}\underset{u^n(a,b)}\sum P_{Z^n,\bm F,U^n}(\cdot,u^n(a,b))\log\left(\frac{P_{Z^n,\bm F|U^n}(\cdot|u^n(a,b))}{\square(z^n,\bm f)}+ 1\right)\notag\\
&\overset{(o)}=\mathds{E}_{Z^n\bm F U^n}\left[\log\left(\frac{P_{Z^n\bm F|U^n}}{L_1Q_{Z^n\bm F}}+1\right)\right],\label{EQ_final_Exp}
\end{align}
where (a) is by constructing a CWTC, such that the key $K$ is interchangeable with the message $M$; (b) is by definition of the divergence \cite[Definition 2.2]{Polyanskiy_LNIT}; (c) is due to the fact that $P_{Z^n,F|M}$ is the marginalization of $P_{Z^n,\bm F|U^n}$ with respect to $w$, which is the index of the confusion message; in (d) we expand the expectation with respect to $M$; (e) is by defining $\sum_{l=1}^{L_1}P_{Z^n,\bm F|U^n}(z^n,\bm f|u^n(m,l))$ and $L_1Q_{Z^n,\bm F}(z^n,\bm f)$ by $\vartriangle_m(z^n,\bm f|u^n)$ and $\square(z^n,\bm f)$, respectively, to simplify the expression; (f) is by definition of the expectation over $\{U^n(m,w)\}_{m=1,w=1}^{L,L_1}$. Since $\{U^n(m,w)\}$ are generated independent and identically distributed using $P_U^n$, the joint distribution of codewords in a codebook is the product of marginal distributions; in (g) we expand the summation with respect to $m$ and $w$; in (h) we expand the product according to the form in step (g); in (i) we collect terms to form the expectation $\mathds{E}_{\underline{U}^n\setminus\underline{U}^n(k,l)}$; in (j) we collect the terms by introducing additional indices $(a,b)$; in (k) we apply Jensen's inequality to the logarithm; (l) is by expanding the expectation $\mathds{E}_{\underline{U}^n\setminus \underline{U}^n(k,l)}$; (m) is by adding the term $P_{Z^n\bm F U^n}(z^n,\bm f,u^n(a,b))$; (n) is by definition of marginalization over $P_{Z^n\bm F U^n}$ with respect to $U^n$. In particular, the 2nd term on the RHS of the numerator in (m) becomes $\mathds{E}_{U^n}[P_{Z^n,\bm F|U^n}]=Q_{Z^n\bm F}$ from \eqref{EQ_target_distr}; (o) is by definition of the expectation.\\

\bibliographystyle{IEEEtran}
\bibliography{IEEEabrv,SecrecyPs22,../SecrecyPs2,../SecrecyPs23}

\end{document}